\documentclass{elsarticle}

\usepackage{setspace}
\usepackage[margin = .9in]{geometry}
\onehalfspace

\usepackage{tikz}
\usepackage{pgfplots}
\usepackage{amssymb}
\usepackage{caption}
\usepackage{subfigure}
\usepackage{amsmath}
\usepackage{amssymb}
\usepackage{array}
\usepackage{graphicx}
\usepackage{tikz-cd}
\usepackage{wasysym}
\usepackage{amsthm}
\usepackage[ruled,noend,noline,linesnumbered]{algorithm2e}
\usepackage{comment}

\usepackage{thm-restate}

\SetKwComment{Comment}{$\triangleright$\ }{}
\SetCommentSty{myargfont}

\graphicspath{{./Images/}}

\usetikzlibrary{matrix}

\newcolumntype{C}[1]{>{\centering\arraybackslash}m{#1}} 

\newtheorem{definition}{Definition}
\newtheorem{theorem}{Theorem}
\newtheorem{lemma}{Lemma}

\newcommand{\TMD}{\mathrm{TMD}}
\newcommand{\PH}{\mathrm{PH}}
\newcommand{\BPH}{\mathrm{EPH}_0}

\pgfplotsset{every axis/.append style={
    axis x line=middle,    
    axis y line=middle,    
    axis line style={<->}, 
    xlabel={$x$},       
    ylabel={$y$},       
    },
    cmhplot/.style={color=red,mark=none,line width=1pt,<->},
    soldot/.style={color=red,only marks,mark=*},
    holdot/.style={color=red,fill=white,only marks,mark=*},
}

\begin{document}

\begin{frontmatter}
\title{Barcodes distinguish morphology of neuronal tauopathy}

\author[add1]{David Beers}
\ead{david.beers@maths.ox.ac.uk}
\author[add2]{Despoina Goniotaki}
\author[add2]{Diane P. Hanger}
\author[add1]{Alain Goriely}
\author[add1]{Heather A. Harrington}

\address[add1]{Mathematical Institute, University of Oxford}
\address[add2]{Institute of Psychiatry, Psychology, \& Neuroscience, King's College London}

\begin{abstract}
The geometry of neurons is known to be important for their functions. Hence, neurons are often classified by their morphology. Two recent methods, persistent homology and the topological morphology descriptor, assign a morphology descriptor called a barcode to a neuron equipped with a given function, such as the Euclidean distance from the root of the neuron.  These barcodes can be converted into matrices called persistence images, which can then be averaged across groups. We show that when the defining function is the path length from the root, both the topological morphology descriptor and persistent homology are equivalent. We further show that persistence images arising from the path length  procedure provide an interpretable summary of neuronal morphology. We introduce {topological morphology functions}, a class of functions similar to Sholl functions, that can be recovered from the associated topological morphology descriptor. To demonstrate  this topological approach, we compare healthy cortical and hippocampal mouse neurons to those affected by progressive tauopathy. We find a significant difference in the morphology of healthy neurons and those with a tauopathy at a postsymptomatic age. We use persistence images to conclude that the diseased group tends to have neurons with shorter branches as well as fewer branches far from the soma.
\end{abstract}
\end{frontmatter}

\section{Introduction}
Neurons are essential for cognitive function, and their activity is dependent on their morphology \cite{london2005dendritic}. For example, when two simultaneous signals reach a neuron's dendrites, their mutual distance affects the combined signal that reaches the soma. In particular, nearby signals have a weaker total affect than those that are far apart \cite{rall1967dendritic}. Similarly, the length of the dendrites is known to affect the signal received by the soma, with signals that have traveled a long distance along the dendrites being weaker and more spread out than those that travel a short distance to the soma \cite{fetz1983relation}. To further confound the study of neurons, it is known that neuronal morphology- even within the same animal \cite{grueber2002tiling}- can be highly heterogeneous, suggesting that different neurons have morphologies suited for different functions.

These observations have naturally lead to attempts to classify large sets of neurons via their geometry. Computationally, neurons are typically represented by trees, collections of vertices in the two or three dimensional Euclidean space connected by linear edges, with one such point denoting the location of the soma. Comparison of large classes of neurons can then be achieved by extracting numerical morphological descriptors, such as the number of bifurcations and branching angles, from these reconstructions \cite{laturnus2020systematic}. However, these individual numerical descriptors cannot describe neuronal morphology in detail (e.g., the number of bifurcations cannot recover a neuron's tortuosity).
Another approach consists in computing density maps that capture the distribution of neurites in a neuron \cite{jefferis2007comprehensive}.
If we consider two classes of neurons with different density maps, it is difficult to determine the neuronal features that correspond to differences in these density maps as well as accurately measure the difference between two density maps (i.e., minimize the integral of their difference over all possible rotations and translations).

A classical morphological descriptor is the \textit{Sholl function} of a neuron \cite{sholl1953dendritic}, which assigns for every positive number~$r$ the number~$s(r)$ of times a spherical shell of radius~$r$ intersects a given neuron. The Sholl function can be interpreted visually when plotted as a function of $r$. By taking averages of Sholl functions, researchers are able to visualize the average structure of large classes of neurons.

 Given a neuron, represented mathematically as a rooted tree, and a function on its nodes,  taken to be the Euclidean distance to the soma, two groups  proposed independently the use of barcodes to compare morphologies \cite{kanari2018topological,li2017metrics}. Barcodes are mathematical objects used in topological data analysis as multiscale morphology descriptors.
For instance, Li et al. \cite{li2017metrics} computed a barcode of neuronal data by applying  \textit{persistent homology} (PH) \cite{carlsson2009topology,edelsbrunner2010computational,ghrist2008barcodes}, a technique in computational mathematics for extracting multi-scale topological features from data. Technically, they compute zero dimensional extended persistent homology, replacing an infinite feature with the maximum value of the function, to obtain a barcode of bounded intervals for a neuron.
Li et al also show that the Sholl function of a neuron can be computed from the union of two particular barcodes. In a separate work, Kanari et al. \cite{kanari2018topological} developed an algorithm that takes a tree and function as input, and calculates a PH inspired barcode, called a \textit{topological morphology descriptor }(TMD). The TMDs of neurons are then converted to matrices called \textit{unweighted persistence images} \cite{adams2017persistence}, to compare different classes of neurons. These unweighted persistence images are visualized as two dimensional images with the intensity of the $(i,j)^{\mathrm{th}}$ pixel representing the magnitude of the corresponding $(i,j)^{\mathrm{th}}$ matrix entry.

Here, we build on the work of both groups. We compute TMDs, but we use  the path distance, or \textit{intrinsic distance} to the soma rather than  the radial (Euclidean) distance from the soma as done in  \cite{kanari2018topological}. The intrinsic distance to the soma has already been suggested as a relevant choice of function \cite{kanari2018topological,li2017metrics} and briefly demonstrated on an example in \cite[SI]{kanari2018topological}; however, to the authors' knowledge it has not been thoroughly investigated for neuronal morphology.

\subsection{Contributions}

We show that different regions in the persistence images of TMDs generated with the path distance to the soma function can be interpreted to understand the morphology of large groups of neurons. Specifically, we discuss how weighted persistence image records information about neuron branches, including their length within a population of neurons as well as their distance from the soma. Based on this framework we have the following contributions.

\begin{itemize}
    \item We define a new class of morphology descriptors, the \textit{topological morphology functions}, and show that they can be approximated from persistence images via numerical integration.
    \item We prove that for a class of functions containing the path distance to the soma 
    the barcodes of \cite{kanari2018topological} and \cite{li2017metrics} are equivalent (Theorem \ref{thm:coresp}). 
    \item We apply these techniques to neurons from control mice and mice that model tauopathies, a form of neurodegenerative disorders in which tau protein forms deposits in the brain. This topological framework finds that neurons in mice with a postsymptomatic tauopathy are on average shorter than controls; furthermore, these neurons exhibit less branching far from the soma than controls
\end{itemize}

The remainder of the paper is organized as follows. In section \ref{sec:theory} we recall the how one constructs the TMD from a neuron equipped with a function, and how to represent a TMD as a persistence image. Section \ref{sec:thrres} describes the morphological information recorded by the TMD and its associated persistence images when a neuron is equipped with the path distance from the soma function. Further, this section exposes how the definition of topological morphology functions arises naturally from these TMDs, and shows how topological morphology functions may be approximately recovered from persistence images. Sections \ref{sec:methods} and \ref{sec:expres} consist of an application of TMDs to study diseased mice. In the appendix we define the barcodes of \cite{li2017metrics}, which are given by persistent homology. Here we also define some technical metrics on the space of persistence diagrams and prove theoretical statements made in Section \ref{sec:thrres}.

\section{Theory}
\label{sec:theory}
\subsection{The topological morphology descriptor (TMD)}
Consider a neuron be represented by a tree $T$ with nodes $N(T)$, and any function $f:N(T) \to [0,\infty)$ which returns a measure of distance from the soma. The \textit{topological morphology descriptor} (TMD) of this neuron equipped with the function $f$ is a collection of intervals in the real line that describes the branching pattern of the dendrites with respect to $f$. To understand how the TMD is computed, we first introduce basic definitions regarding the tree $T$. In this paper, we consider only \textit{finite trees}, i.e. trees with finitely many nodes. We also assume that all trees are equipped with embeddings in $\mathbb{R}^2$ or $\mathbb{R}^3$, sending edges to line segments, unless stated otherwise. Computationally, a neuron is represented by a tree $T$, with nodes associated to locations in $\mathbb{R}^3$, and with a distinguished node~$r$ that represents the soma, as shown in Figure \ref{fig:barcode} (a). We refer to the distinguished node as the \textit{root} of $T$. From the root we induce an orientation of the edges of $T$. We orient any edge $e$ incident to the $r$ away from~$r$. Inductively, if $v$ is incident to an edge we have already oriented, we orient any edges incident to $v$ and $v'$, but not yet oriented, away from $v$ towards $v'$. If, for a given $v$ and $v'$, in the resulting directed graph there is a directed edge $e$ from $v$ to $v'$, denoted $(v,v')$, we say that $v$ is the \textit{parent} of $v'$ and $v'$ is a \textit{child} of $v$. If a vertex $v$ has three or more incident edges, we say that $v$ is a \textit{branch point}. If the root $r$ has two or more incident edges then we say that $r$ is a branch point. Similarly, if a vertex $v \neq r$ has exactly one incident edge, we say that $v$ is a \textit{leaf} of $T$. Suppose that $v'$ is a child of $v$, i.e. there is a directed edge $e = (v,v')$. We can associate to $v'$ the induced subtree $T'$ generated by $v'$, its children, its children's children, and so on. We say that the union $T'\cup \{e\}$ is a \textit{child branch} of $v$. Examples of these definitions are shown in Figure~\ref{fig:extree}. In \cite{kanari2018topological} the function~$f$ on~$T$ is chosen to be the Euclidean distance radial distance from the root $r$.

\begin{figure}[htbp]
\centering
\resizebox{0.6\textwidth}{!}{
\includegraphics{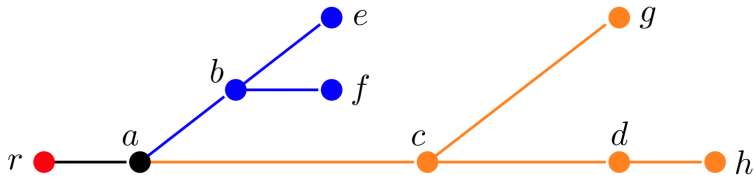}
}

\caption{A tree with root $r$. Nodes $a$, $b$ and $c$ are branch points, while nodes $e$, $f$, $g$, and $h$ are leaves. Nodes $r$ and $d$ are neither. The two child branches of $a$ are colored in blue and orange.}
\label{fig:extree}
\end{figure}

The algorithm of \cite{kanari2018topological} yields a local-to-global topological summary of the morphology of $T$ using the function $f$ as follows. At each branch point $b$ in a tree $T$ we inspect each child branch of $b$. We leave untouched the child branch that attains the greatest value of the function $f$ on its leaves, and detach every other child branch of the branch point. It may be the case that at a given branch point, two or more different child branches attain the greatest value of the function $f$ on their leaves. In such a case, we select, arbitrarily, one child branch to keep connected to the tree and we detach the other branches. This choice has no effect on the output of the algorithm. We then iterate over all branch points. Once this process has been completed at each branch point, what will remain is a collection of intervals, with endpoints given by their associated $f$ values, which we denote by $\TMD(T,f)$. It may be the case that some interval appears more than one time in $\TMD(T,f)$. When this happens we record the number of times that the interval appears. We show this process applied to an example neuron in Figure~\ref{fig:barcode}. Kanari et al \cite{kanari2018topological} choose a specific order in which their algorithm inspects the branch points; however, the chosen order does not change the resulting TMD, since applying the detaching procedure at one branch point has no effect on neither the number of child branches any other branch point has, nor the maximal leaf values of $f$ on these child branches. In general, a multiset of intervals with endpoints indexed by real numbers, such as $\TMD(T,f)$, is referred to as a \textit{barcode}. We summarize the procedure of generating $\TMD(T,f)$ in the following steps:
\begin{enumerate}
\item Choose a branch point, and identify a child branch of that branch point that attains the greatest value of $f$ on its leaves.
\item Detach every child branch from this branch point except for the identified branch.
\item If there is branch point in the resulting collection of trees, return to step~1.
\item Label the endpoints of the resulting collection intervals with their associated $f$ values.
\end{enumerate}

\begin{figure}[htbp]
\centering
\resizebox{\textwidth}{!}{
\includegraphics{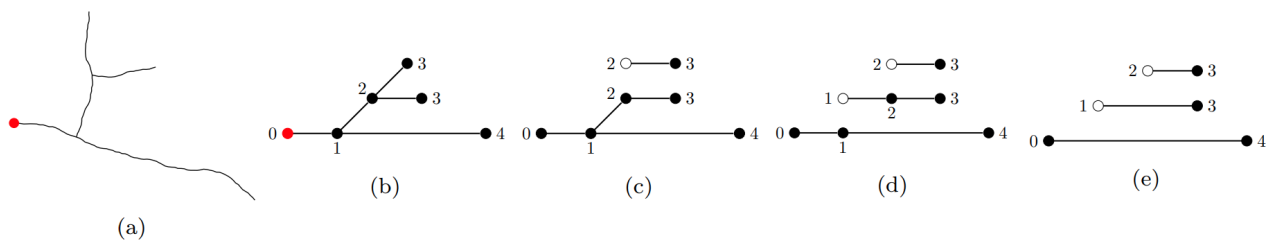}
}

\caption{Retrieving the barcode from a neuron. In (a) is a simple computer generated neuron, with its soma identified by a red node. Figure (b) depicts the neuron represented as a tree with its root colored red. In (b) we also label each node by the values of a function $f$. Panels (b) through (e) show the process of extracting a barcode from a rooted tree with labelled edges. In the transition from (b) to (c), we remove a branch at the branch point indicated by the number 2. Here, it does not matter which branch we remove since both branches attain the same maximal value of 3 on their leaves. In the transition between (c) and (d) we remove a branch at the branch point labelled 1. This time we do not have any choice of which branch to remove: one branch has a leaf with a value of 4, so we must remove the branch that only has a leaf with a value of 3. Panel (e) shows the barcode we acquire after this process. Each interval corresponds to a branch in the original neuron, with its left endpoint given the $f$ value at the branch point where the branch initiates, and the right endpoint given the $f$ value at the leaf where the same branch terminates.}
\label{fig:barcode}
\end{figure}

Notably, there is a bijective correspondence between right endpoints of intervals and leaves of $T$. Hence the right endpoints must be closed. All but one of the left endpoints of $\TMD(T,f)$ are constructed by detaching a child branch from a branch point, and so must be open. The only interval with a closed left endpoint must be $[f(r),L]$, with $L$ the maximum value of $f$ on the leaves. To see this, consider the unique path along directed edges from $r$ to a leaf $l$ with $f(l) = L$. At each branch point on this path we can choose not to detach the child branch containing $l$ during the TMD algorithm. Therefore the interval $[f(r),L]$ will be one of the remaining intervals after the TMD algorithm terminates. For an arbitrary function $f$, the intervals $[x,y]$ or $(x,y]$ of the $\TMD(T,f)$ are not always technically intervals in the real line, since it can be the case that $x>y$. However, outside of the appendix we will restrict our attention to functions $f$ such that the intervals have the property $x<y$.

\subsection{Methods for analyzing TMDs}
A barcode can be visualized as multisets of intervals (see Figure \ref{fig:barcode}).
Barcodes can also be visualized as multisets of points in the real plane by sending each interval $[x,y]$ or $(x,y]$ to the point $(x,y)$. We call this multiset a
\textit{persistence diagram}. 
We remark that both barcodes and persistence diagrams record the \textit{multiplicities} of points in the plane (i.e., the number of intervals with the same endpoints), and these multiplicities can be visualized (see e.g., \cite[Figure 2]{hiraoka2016hierarchical}). To be able to compute distances between persistence images, every point $(x,x)$ is included in every persistence diagram with infinite multiplicity, and are typically either not depicted visually or represented by a diagonal line. The bottleneck distance and the $q$-Wasserstein distances (both defined in the appendix) are two frequently used distance functions between persistence diagrams which make use of these diagonal points. However for our purposes the diagonal points are not morphologically relevant. In Figure \ref{fig:persim}, we show the example neuron from Figure $\ref{fig:barcode}$, its barcode, and its persistence diagram in panels (a), (b), and (c) respectively. The height of each point $(x,y)$ above the line $y=x$, dashed in Figure \ref{fig:persim} (c), is $y-x$, which is also the difference in $f$ values of the endpoints of the corresponding interval in the barcode. By convention, we refer to $x$ as the \textit{birth} of the interval, $y$ as the \textit{death} of the interval, and the length $y-x$ as its \textit{persistence}. The persistence diagram of $\TMD(T,f)$ is known to be stable (with respect to the bottleneck distance) to perturbations of $f$, addition of short branches to $T$, and removal of short branches from $T$ \cite[SI,~Theorem~1]{kanari2018topological}.

Since the persistence of an interval is an important quantity, a natural transformation to apply is
\begin{equation}
(x,y) \mapsto (x, y-x).
\end{equation}
In Figure \ref{fig:persim} (d) we show the above transform applied to the persistence diagram in Figure \ref{fig:persim} (c). This transformation leaves the birth value of each point unchanged, but replaces the death value of a point with its persistence. Note that this transformation is invertible with inverse $(x,y) \mapsto (x,x+y)$, and so does not lose any information from the persistence diagram.

\begin{figure}[htbp]
\centering
\resizebox{\textwidth}{!}{
\includegraphics{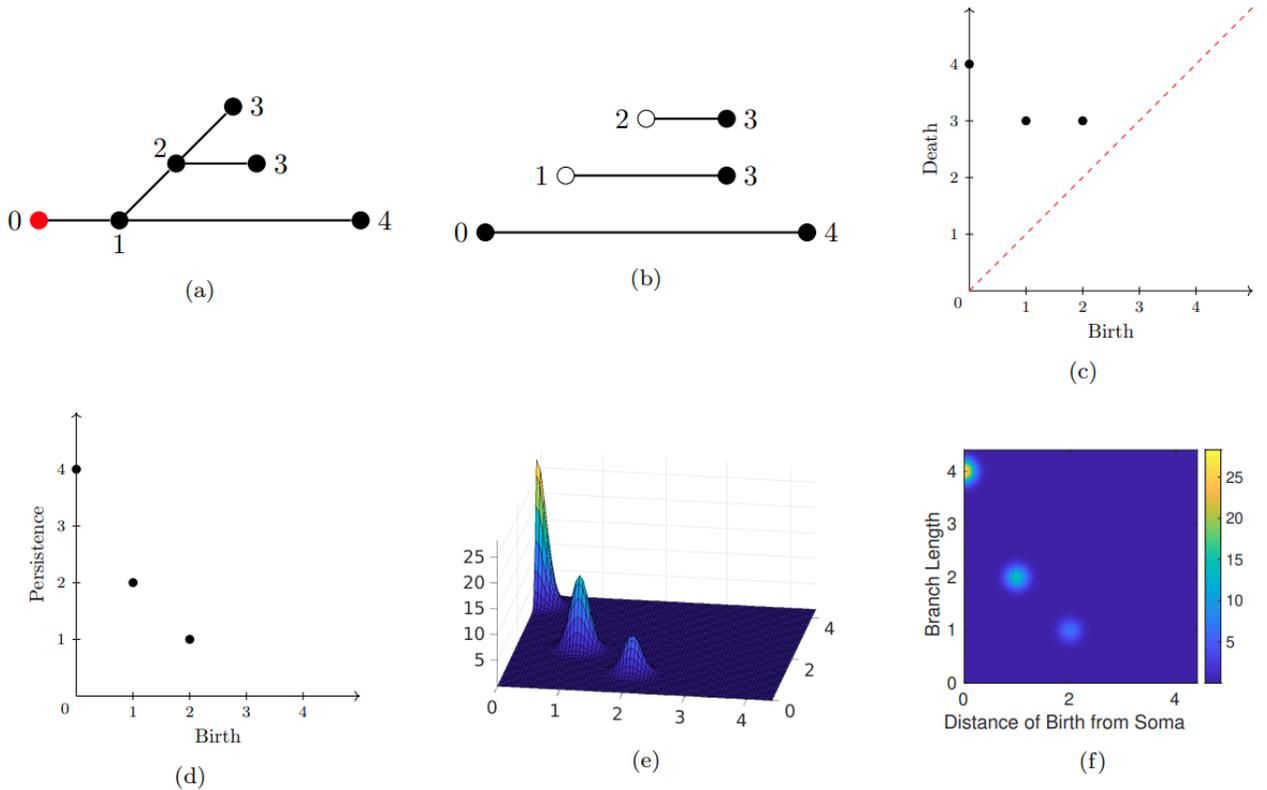}
}

\caption{Turning a tree into a persistence image. Panel (a) shows an example tree with nodes labelled by their path length distance to the root, which is colored red. Panels (b) and (c) show the barcode diagram and persistence diagram of this neuron respectively. Panel (d) shows the resulting transformed persistence diagram. Lastly, panels (d) and (e) show the resulting persistence surface and persistence image respectively. Note that the points with greater $y$ values are weighted more heavily in the construction of panels (e) and (f).}
\label{fig:persim}
\end{figure}

While barcodes, diagrams, and their transformed counterparts can be convenient to visualize, they are not suitable for averages and statistical comparison.
One solution proposed by Adams et al \cite{adams2017persistence}
is to define a two-dimensional function that is the sum of Gaussian functions of a fixed standard deviation $\sigma$ centered at each point in the transformed (birth, persistence) diagram. Each Gaussian function involved in this sum is given a weight equal to the height above the $x$-axis of the corresponding point. Such a function associated to a barcode is called a \textit{persistence surface} \cite{adams2017persistence}. Figure~\ref{fig:persim}(e) shows the persistence surface for the example neuron. As a result of the choice of weighting, the area under each such function is proportional to the sum of the lengths of intervals in $\TMD(T,f)$. In practice, persistence surfaces are often simplified by considering the value of the persistence surface on only a finite grid of values. Such a simplification is referred to as a \textit{persistence image}\footnote{Weighting the Gaussian functions is standard practice when generating a persistence image. As the authors of \cite{kanari2018topological} skip the weighting step, they refer to persistence images of the TMDs they study as unweighted persistence images. The authors also do not transform their data prior to generating persistence images, since intervals can have negative persistence in an arbitrary TMD. Since we will restrict to barcodes with intervals of only positive persistence in the methods section, we transform the persistence diagram first, which is also standard.} \cite{adams2017persistence}. Computationally, persistence images can be represented by matrices, with the $(i,j)^{\textrm{th}}$ entry being the value of the persistence surface at row $i$ and column $j$ of the chosen grid. Most commonly, persistence images are displayed as heat maps, as in Figure \ref{fig:persim} (f). An advantage of analyzing persistence images instead of barcodes or diagrams is the ability to efficiently compute differences and averages across of sets of neurons. 

While the distance of persistence images is not stable with respect to the bottleneck distance on the space of barcodes, it is stable with respect to the 1-Wasserstein distance \cite[Theorem~4]{adams2017persistence} thanks to the weights applied to the Gaussian functions.

\section{Theoretical Results}
\label{sec:thrres}
\subsection{Interpretation}

We restrict our attention to barcodes given by $\TMD(T,d)$, where $d(v)$ is the intrinsic distance from $r$ to $v$ in $T$ as a subset of $\mathbb{R}^3$. In other words, the function $d$ returns the length of the path along the neuron from the point given by $v$ to the soma. For example, $d$ is used to obtain the TMD in Figure \ref{fig:barcode}. We show that the barcode calculated using the methodology of Li et al. \cite{li2017metrics} is equivalent to $\TMD(T,d)$ (see Theorem \ref{thm:coresp}, appendix). Since the TMD algorithm decomposes $T$ into a collection of branches, the sum of the lengths of all the intervals gives the total branch length of the neuron. Further, the only closed interval in $\TMD(T,d)$ will be $[0,L]$, where $L$ is the greatest value of $d$, since $d$ attains all of its local maxima on leaves and $d(r) = 0$. All other intervals $(x,y]$ represent a branch in the barcode decomposition of $T$ that initiates at an intrinsic distance $x$ from the root and terminates at an intrinsic distance of $y$ from the root. Therefore, the persistence of each interval in $\TMD(T,d)$ is exactly the length of the branch that the interval represents, which must be positive. Hence our choice of weights ensures the area under the persistence surface of $\TMD(T,d)$ is exactly the cumulative branch length of~$T$.

A persistence surface of $\TMD(T,d)$ can be analyzed and interpreted back to its associated neuron in other ways as well.
By identifying regions in the $xy$-plane where persistence surface $F(x,y)$ is large, we can read off the types of branches in $\TMD(T,f)$ that contribute to the total branch length. Indeed, 
regions where $F$ is large with large $y$ values correspond to longer branches, while regions where $F$ is large with small $y$ values correspond to the contribution of shorter branches.
Similarly, regions where $F$ is large with large $x$ values correspond to branches that initiate close to the root. By contrast, regions where $F$ is large with small $x$ values correspond to branches that initiate far from the root, in the sense that the path from the root to the branch point at which they initiate is long.
Note that the persistence image will always have a Gaussian contribution of weight $L$ centered at $(0,L)$, since $\TMD(T,d)$ must contain the interval $[0,L]$. Since persistence images are discrete approximations of persistence surfaces, they inherit similar interpretations. These concepts are illustrated with artificial neurons, each with the same total branch length, in Figure \ref{fig:toy_barcode}. In Figure \ref{fig:drospersim} we show how persistence images can summarize the traits of real morphologically distinct neurons.
As mentioned, an advantage of persistence images over persistence diagrams is that it is straightforward to compute and visualize the average persistence image of a collection of neurons. From such a persistence image we can then easily calculate the average total branch length, which is given by the discrete approximation of the integral over the

persistence image. We can also identify the types of branches that tend to appear in an entire class of neurons.

\begin{figure}[htbp]
\centering
\resizebox{\textwidth}{!}{
\includegraphics{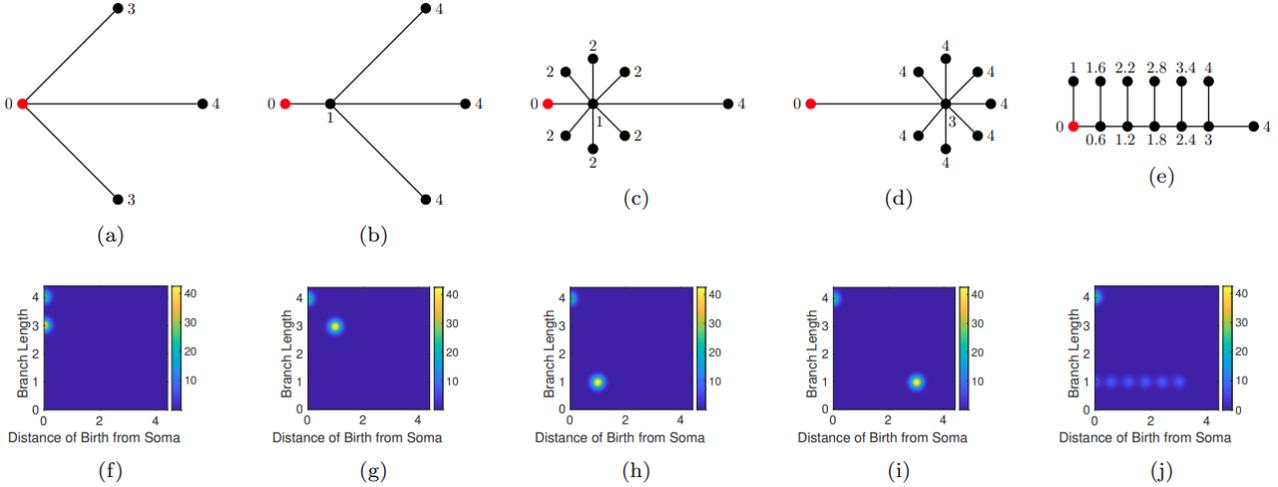}
}

\caption{ Persistence images of toy neuronal trees. Top row: Five toy neuronal trees with roots indicated by a red vertex. Bottom row: Persistence images corresponding to the $\TMD(T,d)$ of the toy neuronal tree in top row.  All persistence images of these trees have a feature centered at $(0,4)$, indicating that the furthest path length from the root (soma) of the toy neurons is 4 units. The remaining pixels with positive values correspond to the remaining branches of each toy neuron. (a) There are two additional branches of total length 6 that originate at the soma, which correspond to a feature centered at the coordinates (0,3) in (f). The relative pixel intensities reflect that these branches correspond to a total length of 6 while the longest branch only corresponds to a total length of 4.  (b) Two shorter branches are a distance of 1 unit from the soma, which corresponds to a birth value of 1 and branch length 3 in (g). (c) Six branches of length 1 all initiate at a distance of 1 from the soma, which corresponds to birth and persistence value of 1. In (d) and (i), six branches of length 1 that initiate at distance 3 from the root causes an increase in pixel intensity near the coordinates (3,1). (e) and (j), Six branches of length 1 initiate at different distances from the root which correspond to positive pixels along the line with a $y$ value of one, however, since they initiate at different distances from the root, the pixel intensity is dispersed across birth values.}

\label{fig:toy_barcode}
\end{figure}

\begin{figure}[htbp]
\centering
\resizebox{\textwidth}{!}{
\includegraphics{Figure_5}
}

\caption{Persistence images of \textit{Drosophila} neurons. (a) A class I neuron characterized by few main branches and many long secondary branches growing perpendicularly to the main branch. (b) An area-covering class IV sensory neuron with many small branches far from the root. In both (a) and (b) the soma is indicated by a red node. Below each in (c) and (d) are the corresponding persistence images.
There are several bright regions substantially above the $x$-axis in (c) resulting from the many long secondary branches in (a). The brightest region of (d) is far to the right and close to the $x$-axis since many of the branches of (b) are far from the root and short. The digital reconstructions of these neurons from \cite{sulkowski2011turtle} are freely available at NeuroMorpho.org \cite{ascoli2006mobilizing}. These neurons are named 02-16-09-Class1-B40X and 02-16-09-ClassIV on the site respectively.}
\label{fig:drospersim}
\end{figure}

\subsection{Proposing topological morphology functions}
In \cite{li2017metrics}, the authors show that a neuron's Sholl function can be reconstructed from a barcode obtained via an alternate methodology. We will show that a similar function can be recovered from the barcode $\TMD(T,d)$. With the function $d$ we lose information about the radial distance of branches from the soma necessary for the construction of Sholl functions. However, we can construct another family of functions which similarly describe the branching morphology of a given neuron. Explicitly, the \textit{topological morphology function} $p$ of a neuron associates to each positive number $t$ the number of points on the neuron which have an intrinsic distance of $t$ to the soma. The value $p(t)$ is the number of intervals in $\TMD(T,d)$ containing $t$, since this is the number of branches in the TMD decomposition of $T$ containing points a distance of $t$ from the soma. We call a number $t$ \textit{generic} if it is not an endpoint of an interval in $\TMD(T,d)$. For generic $t$, the number $p(t)$ is easily obtained from both persistence diagrams and transformed persistence diagrams, as shown in Figure~\ref{fig:trianglecount}. In a persistence diagram, $p(t)$ for a generic $t$ is the number of points above and to the left of $(t,t)$. In a transformed persistence diagram, the same value is the number of points in the region $R_t$, given by the equations
\begin{equation}
x < t, \quad y > t-x.
\end{equation}
The reason why we can only deduce $p(t)$ from persistence diagrams for generic $t$ is that persistence diagrams do not retain information about the openness or closedness of the intervals in their corresponding barcode.

\begin{figure}[htbp]
\centering
\resizebox{.8\textwidth}{!}{
\includegraphics{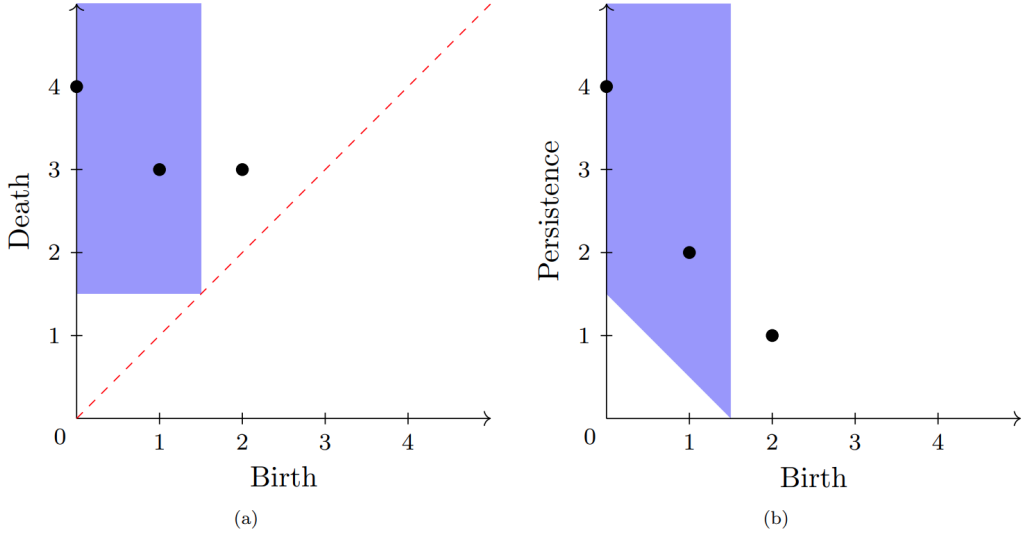}
}

\caption{Recovering $p(1.5)$ from a standard and transformed persistence diagram. On the left, $p(1.5)$ is recovered from the persistence diagram by counting the number of points above and to the left of the point~$(1.5,\,1.5)$. This region is shaded blue. On the right, $p(1.5)$ is recovered from the corresponding transformed persistence diagram by counting the number of points in $R_t$, shaded blue. In both cases, the number of points in the blue region is 2, so $p(1.5) = 2$.}
\label{fig:trianglecount}
\end{figure}

\noindent\begin{minipage}{\textwidth}
The following theorem shows that the topological morphology function $p$ of a neuron represented by $T$ can be partially reconstructed from the persistence surface of $\TMD(T,d)$ via integration.

\begin{restatable}{theorem}{limfunc}
Let $p$ be the topological morphology function associated to a neuron represented by a tree $T$. Let $d:N(T) \to [0,\infty)$ be the intrinsic distance to the soma. Let $F_\sigma(x,y)$ be the persistence surface corresponding to $\TMD(T,d)$ constructed with Gaussian functions of standard deviation $\sigma$. For any generic positive number $t$,

\begin{equation}
p(t) = \lim_{\sigma \to 0} \int_{R_t}\frac{F_\sigma(x,y)}{y}\,dx\,dy.
\label{eqn:limfunc}
\end{equation}
Further, the integral in the above expression converges and is an infinitely differentiable function of $t$ for all $\sigma > 0$.
\label{thm:limfunc}
\end{restatable}
\vspace{\parskip}
\end{minipage}

This theorem shows that we may use the persistence surface of a barcode to retrieve a smooth approximation of the topological morphology function, which improves in accuracy as the standard deviation~$\sigma$ used to generate the persistence surface approaches zero. We show in Figure \ref{fig:persfapprox} that this is true in when we approximate integration with the persistence image.
Note that if every branch of a neuron grows in a straight line exactly radially outward from the soma, then the topological morphology function of the neuron is equal to the neuron's Sholl function~$s$. Since this is often a reasonable approximation, we expect that $p$ and $s$ are similar in many cases. Under such an assumption regarding the direction of neurite outgrowths, this theorem shows that we can use a persistence image to approximate smoothly the Sholl function $s$.

We defer the proof of Theorem \ref{thm:limfunc} to the appendix, where we also prove an analogous result for the barcodes of \cite{li2017metrics}.

\begin{figure}[htbp]

\centering
\resizebox{0.5\textwidth}{!}{
\includegraphics{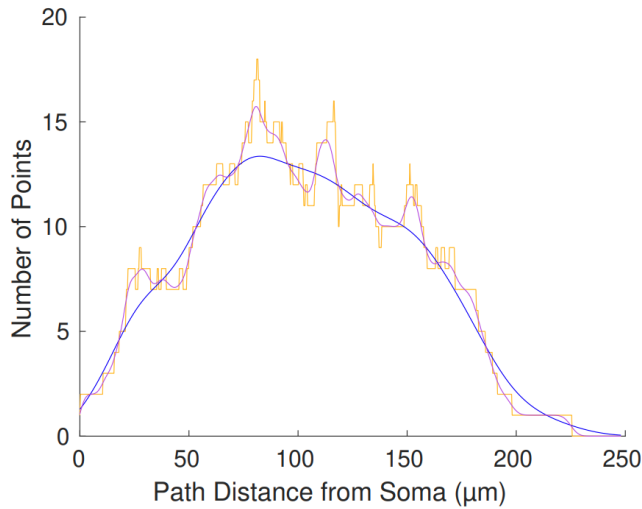}
}

\caption{The topological morphology function (orange) for the neuron in Figure \ref{fig:drospersim} (a) along with the approximate topological morphology function obtained by numerical integration of the persistence image for $\sigma = 10\mu m$ (blue) and $\sigma = 2 \mu m$ (purple).}
\label{fig:persfapprox}
\end{figure}

\section{Methods}
\label{sec:methods}
\subsection{Data}
 
We study Tau35 mice, a transgenic mouse line which has been used as a model to reproduce biological and cognitive features of human tauopathies \cite{bondulich2016tauopathy}. In particular, we  analyze topologically a set of 317 digitally reconstructed mouse neurons,  171 wild type (WT) and 146 Tau35 mouse variant, to study the effect of  tauopathy on morphology. These reconstructed neurons are further divided into groups by age (e.g., presymptomatic 4 month neurons or postsymptomatic 10-12 month neurons) and type (e.g., cortical or hippocampal neurons). All neurons in the dataset were grown in live mice and measured via a staining technique. Each reconstructed neuron features a collection of points representing the boundary of the soma. Since our methods do take into account the morphology of the soma, we preprocess the data by contracting these boundary points to a single point at their center of mass.

\subsection{Compuations}

We compute persistence images of $\TMD(T,d)$ associated to each neuron in the data set and take in-group averages. Each persistence image is generated from the persistence surfaces\footnote{For computations we generated persistence images with $231\times231$ entries. The entries of these persistence images evenly cover the region $[-0.15L_{\textrm{max}},L_{\textrm{max}}]^2\subseteq \mathbb{R}^2$, where $L_{\textrm{max}}$ is 1.1 times the maxiumim in-group branch length, across Tau35 and WT neurons. Note that the persistence images we use for compuation sample values outside the first quadrant, to assure that relevant information of points on the associated persistence diagram near the $x$ or $y$ axes is captured.} of Gaussian functions with a standard deviation of 10$\mu m$. We hypothesize that the difference between group averages of WT and Tau35 neurons of the same age and type is significant. To test this hypothesis, we compute the $L^1$ distance between average persistence images of the Tau35 and WT neurons in question. We then compare this $L^1$ distance to the $L^1$ distances for averages of 1000 randomized regroupings of the data.

\section{Results on tauopathy vs WT neurons}
\label{sec:expres}
 To test the hypothesis that WT neurons and tauopathy neurons are on average morphologically distinct, we average the persistence images of each group of neurons and perform a randomization test as explained in the Methods.
We show the results of our randomization tests in Table \ref{fig:randomizationtable}, where the rightmost column is the percent of times the randomized $L^1$ distance was less than or equal to the same distance between the average Tau35 and WT neurons. For each of the four randomization tests we choose to accept the between-average difference as significant if less than $5/4$ percent of the regroupings had a greater between-average distance. Under this criterion, the between-average difference of the persistence images of Tau35 and WT neurons was significant for the 10-12 month groups, but not significant for the 4 month groups. In Figure \ref{fig:wttaupersim}, we show the average persistence images of the Tau35 and WT neurons of each group with their in-group differences.

\begin{table}
\caption{The counts of neurons of different types in the dataset. The rightmost column denotes the percentage of random regroupings which had an $L^1$ distance of group averages less than the $L^1$ distances of the WT and Tau35 averages. }
\begin{center}
\begin{tabular}{ |c|c|c|c|c| } 
 \hline
 Age(Months)& Region & WT Count & Tau35 Count & Percent\\ 
\hline
\hline
 10-12  & Cortex & 65 & 43 & 99.2\\ 
\hline
 10-12  & Hippocampus & 38 & 49 & 99.7\\ 
\hline
 4  & Cortex & 39 & 33 & 98.6\\ 
\hline
 4 & Hippocampus & 29 & 21 & 94.9\\ 
 \hline
\end{tabular}
\end{center}

\label{fig:randomizationtable}
\end{table}

\begin{figure}[htbp]
\centering
\resizebox{\textwidth}{!}{
\includegraphics{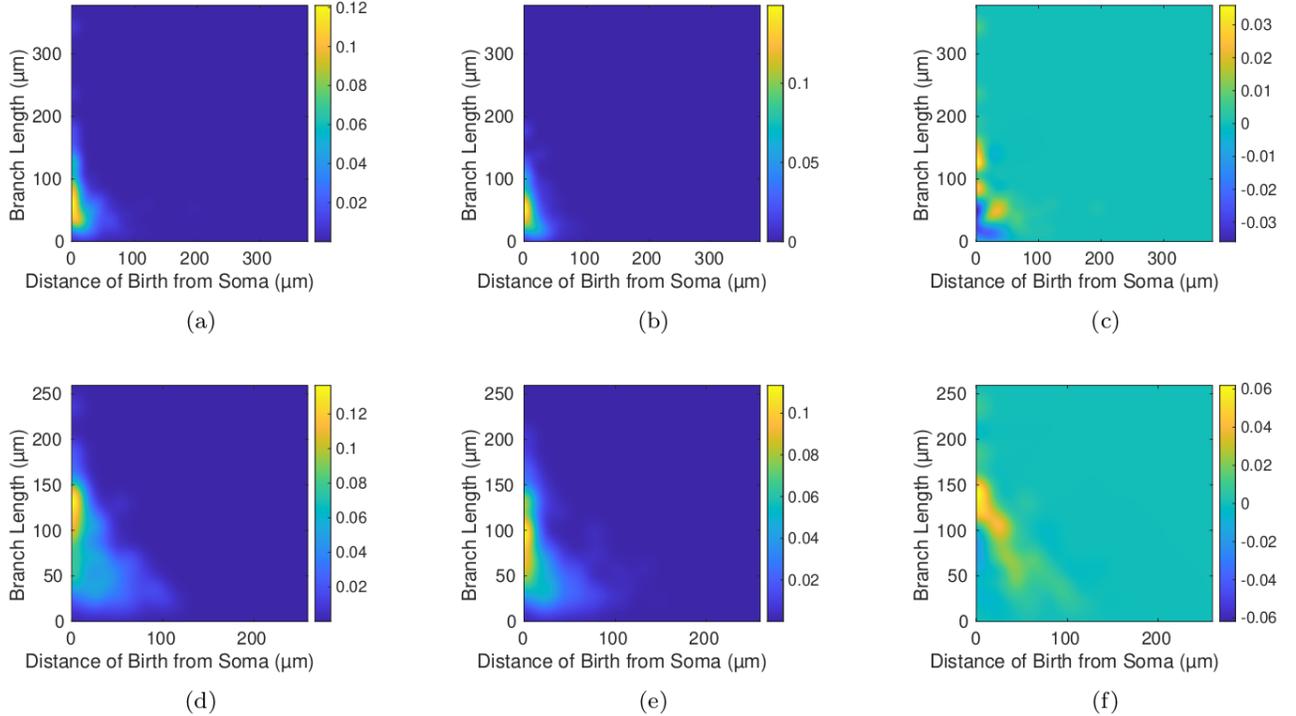}
}

\caption{In-group average persistence images and the difference between Tau and WT groups for 10-12 month cortical (first row) and hippocampal (second row) neurons. (First column) average WT persistence images, (second column) average Tau35 persistence images, (third column) difference between WT and Tau35 average persistence images.}
\label{fig:wttaupersim}
\end{figure}

We analyzed  in-group averages of persistence images to identify features which distinguish WT neurons from Tau35 neurons on average at 10-12 months. For both hippocampal and cortical neurons in this age range (Figure \ref{fig:wttaupersim} (c) and (f) respectively), we observe that the average Tau35 persistence images have a greater intensity towards the origin, while WT persistence images have greater pixel intensity for larger $x$ and $y$ values. From these observations we deduce that branches tend to not only be longer in WT neurons, but also often initiate further along the neuron than in Tau35 neurons.

In Figure \ref{fig:neurfuncs} we plot the associated Sholl and topological morphology functions of the Tau35 and WT 10-12 month cortical and hippocampal neurons. We observe from both the Sholl and persistence functions that there are less branches far from the root in Tau35 neurons than WT neurons. The observation that the topological morphology and Sholl functions are increasing at greater distances from the soma for WT neurons, suggests that branching tends to occur at greater distances from the root in WT neurons. The analysis of the persistence functions, which capture the correlation between branch initiation and branch termination confirms this conclusion. 

\begin{figure}[htbp]
\centering
\resizebox{0.8\textwidth}{!}{
\includegraphics{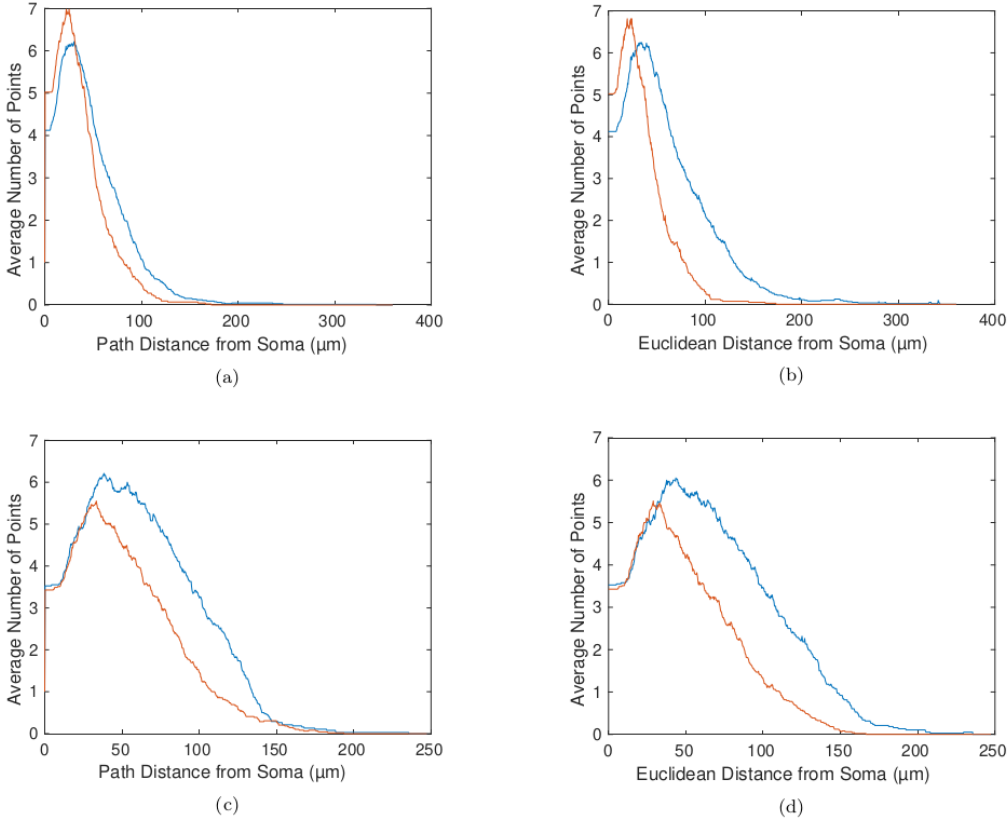}
}

\caption{Average (left column) topological morphology and (right column) Sholl functions for (top row) 10-12 month cortical and (bottom row) 10-12 month hippocampal neurons of both WT (blue) and Tau35 (orange) classes.}
\label{fig:neurfuncs}
\end{figure}
\section{Conclusion}

Barcodes and persistence images are powerful tools for the analysis of tree-like geometries in biological systems, such as neurons. An advantage of persistence images is that they can be used to study statistically groups of neurons. Here, we have further shown that  persistence images based on a notion of path length capture key features of neuronal morphologies and can be used to  discriminate statistically distinct classes of neurons. Once two groups of neurons have been shown to be statistically distinct, the difference between the two groups can be determined through visual inspection of their respective averaged persistence images and their differences. When applied to neurons in healthy and diseased groups, we have shown that the quantitative features of the branches of a neuron or a group of neurons can be easily interpreted from our persistence images.

The natural choice of path length from the soma as an input to the topological morphological descriptor leads to the definition of a topological morphology function similar to the Sholl function. Smooth approximations of this function can be obtained directly from persistence images and in the particular case when branches mostly grow away from the soma, this topological morphology function produces an approximates the Sholl function. By presenting an alternative function for the TMD, we are able to precisely connect the TMD algorithm to persistent homology and
contribute an interpretable descriptor of intrinsic neuronal morphology that complements the toolkit of neuronal analysis.

\section{Acknowledgements}
DB thanks Jacob Leygonie for helpful discussions. DB and HAH are grateful to the support provided by the UK Centre for Topological Data Analysis EPSRC grant EP/R018472/1. DB is grateful for the support of the Mathematical Institute of Oxford. HAH gratefully acknowledges funding from EPSRC EP/R005125/1 and EP/T001968/1, the Royal Society RGF$\backslash$EA$\backslash$201074 and UF150238. AG is grateful for the support by the Engineering and Physical Sciences Research Council of Great Britain under research grants EP/R020205/1. This work was supported by grants from the Alzheimer’s Society and the Alzheimer’s Research UK King’s College London Network Centre to DG and DPH.

\clearpage
\bibliographystyle{abbrv}
\bibliography{refs}

\clearpage
\appendix
\section*{Appendix}
\addcontentsline{toc}{section}{Appendices}
\renewcommand{\thesubsection}{\Alph{subsection}}
\renewcommand\thefigure{A\arabic{figure}}
\renewcommand{\theequation}{A.\arabic{equation}}
\setcounter{figure}{0}
\subsection{Zero dimensional persistent homology for graphs}
This section of the appendix constitutes a brief outline of the concepts from persistent homology that are necessary for other sections of the appendix. We refer the interested reader to \cite{Otter} and \cite[Chapter VII]{edelsbrunner2010computational} for a more complete overview.

Suppose we have a finite connected graph $G$ and a map $f:N(T) \to \mathbb{R}$. We can define the \textit{sublevel sets} $G_t$ to be the subgraphs of $G$, generated by all of the nodes $v$ satisfying $f(v) \leq t$. Any function on $G$ has only finitely many distinct sublevel sets $G_t$, each of which has a lowest corresponding $t$ value $t_i$. Hence we have a nested inclusion of graphs
\begin{equation*}
\varnothing \subseteq G_{t_1} \subseteq G_{t_2} \subseteq \ldots \subseteq G_{t_n} = G,
\end{equation*}
which we can also write as
\begin{equation}
\varnothing \rightarrow G_{t_1} \rightarrow G_{t_2} \rightarrow \ldots \rightarrow G_{t_n} = G,
\label{eqn:graphseq}
\end{equation}
with each arrow denoting the inclusion map. Each graph $G_{t_i}$ has an associated vector space $C_0(G_{t_i})$ given by the $\mathbb{R}$-span of the nodes of $G_{t_i}$. The inclusion maps in the sequence of $G_{t_i}$ then induce linear maps in the sequence
\begin{equation*}
\{0\} \rightarrow C_0(G_{t_1}) \rightarrow C_0(G_{t_2}) \rightarrow \ldots \rightarrow C_0(G_{t_n}) = C_0(G).
\end{equation*}
We can quotient each vector space $C_0(G_{t_i})$ via the relation $v \sim v'$ whenever $v$ and $v'$ are in the same connected component in $C_0(G_{t_i})$. We call the resulting vector spaces $H_0(G_t)$, the \textit{zero dimensional homology} of $G_t$. If $v$ and $v'$ are both in the same connected component of $G_t$, it follows that they are also both in the same connected component of $G_a$ for any $a\geq t$. Hence the inclusion maps above induce inclusion maps between homology groups, giving us the sequence
\begin{equation*}
\{0\} \rightarrow H_0(G_{t_1}) \rightarrow H_0(G_{t_2}) \rightarrow \ldots \rightarrow H_0(G_{t_n}) = H_0(G).
\end{equation*}
This sequence connected by the linear maps induced by inclusion is called the \textit{dimension zero persistent homology module} of $G$ with respect to $f$. We would like to decompose this sequence of vector spaces into simpler sequences of form
\begin{equation*}
    \{0\} \to \ldots \to \{0\} \to \mathbb{R} \to \ldots \to \mathbb{R} \to \{0\} \to \ldots \to \{0\}
\end{equation*}
with identity maps between the spaces $\mathbb{R}$ and zero maps elsewhere. We refer to such a sequence of vector spaces with $\mathbb{R}$ first appearing at the location for $t_i$ and last appearing at the location of $t_j$ as $I(t_i,t_{j+1})$, and we call this type of sequence an \textit{interval module}. For when $j = n$ in this definition, we define $t_{n+1} = \infty$. It was shown in \cite{zomorodian2005computing} that if $\mathcal{M}$ is a zero dimensional persistence module of $G$ with respect to $f$, then $\mathcal{M}$ decomposes uniquely into a direct sum of interval modules
\begin{equation*}
    \bigoplus_{k=1}^K I(t_{i(k)},t_{j(k)}).
\end{equation*}

From this decomposition we associate a barcode, a collection of intervals in the real line, by the mapping
\begin{equation*}
    \bigoplus_{k=1}^K I(t_{i(k)},t_{j(k)}) \mapsto \bigsqcup_{k = 1}^K [t_{i(k)}, t_{j(k)+1})
\end{equation*}
We refer to this barcode as the \textit{persistent homology} of $G$ and $f$, or $\PH(G,f)$. It follows from this decomposition that there are exactly as many infinite intervals, intervals with right endpoint $t_{n+1} = \infty$, in $\PH(G,f)$ as there are connected components in $G$. Since we assume $G$ to be connected, exactly one of these intervals is infinite. Sometimes it is useful to make the resulting barcode only consist of bounded intervals, and we can achieve this, for example, by sending the one instance of $\infty$ to $\max(f)$ and replacing the open endpoint with a closed endpoint. We call the resulting barcode $\BPH(G,f)$.  Readers familiar with extended persistence \cite{cohen2009extending} will note that this is exactly the zero dimensional extended persistence barcode of $G$, justifying the notation.

It is well known that there is a bijection between intervals in $\PH(G,f)$ and the local minima of $f$. These are the connected subgraphs of $G$ induced by sets of nodes on which $f$ takes exactly one value which is less than its value on all neighboring nodes. Suppose all local minima of $f$ are nodes. The values $f(v)$ for minima are the values of the left endpoints of their associated intervals, since they correspond to $t$ values where their associated connected components first appear. Suppose that at time $t_i$ in the sequence given by Equation~\ref{eqn:graphseq} two or more connected components $\{C_l\}_{l=1}^N$ of $G_{t_{i-1}}$ merge. Additionally, suppose there exists unique global minima of $f$ in each $C_l$, denoted $v_l$, and one of the $v_l$ has the greatest $f$ value, without loss of generality $f(v_l) < f(v_1)$ for $l\geq 2$. In this generic case, then the Elder Rule \cite[Theorem 4.4]{cai2021elder} determines that the right endpoint of the interval corresponding to $v_l$ to be $f(t_l)$ for $l\geq2$, giving rise to the interval $[f(v_l),f(t_i))$. If our assumptions hold whenever connected components merge, we can compute the persistent homology by this procedure. Indeed, after applying this process to every merging of connected components in Equation~\ref{eqn:graphseq}, every local minimum whose corresponding right endpoint has yet to be assigned must be the global minimum of its connected component in $G$. Since $G$ is connected, there must only be one such local minimum. Hence, the value of this local minimum must be paired with an infinite right endpoint, as $H_0(G)$ is a one dimensional vector space. In \cite{li2017metrics}, barcodes for a tree $T$ and a function $f$ are obtained by computing $\PH(T,f)$ by this procedure under the assumption of the necessary genericity conditions, and transforming the result to the barcode of bounded intervals $\BPH(T,f)$.

We can easily define define the negation operation on barcodes to be the map which swaps and takes the negative value of interval endpoints. For example, we have
\begin{equation*}
    -\big( [1,\infty) \sqcup [2,5) \sqcup [3,4)\big) = (-\infty,-1] \sqcup (-5,-2] \sqcup (-4,-3]
\end{equation*}
Similarly, we can also define a switching map $S$ which switches the endpoints of each interval. For example,
\begin{equation*}
    S\big( [1,\infty) \sqcup [2,5) \sqcup [3,4)\big) = (\infty,1] \sqcup (5,2] \sqcup (4,3].
\end{equation*}
The switching map does not send intervals in $\mathbb{R}$ to intervals in $\mathbb{R}$ since the right endpoint will now be less than or equal to the left endpoint of every interval. Nevertheless, this is a useful notion to consider.

In the sequence of Equation $\ref{eqn:graphseq}$ we constructed $G$ sequentially via the sublevel sets of $f$. We could proceed analogously using the \textit{superlevel sets} of $f$, the subgraphs $G^t$ generated by the nodes $v$ satisfying $f(v) \geq t$, this time with $t$ decreasing with each inclusion map. The exact same procedure gives us a barcode corresponding to the sublevel sets of $f$, with the caveat that the so-called intervals $[x,y)$ in this barcode satisfy $x \geq y$. Similarly, we can make these intervals finite by replacing the of $-\infty$ in the resulting barcode with $\min(f)$, and replacing the open endpoint with a closed endpoint. It is readily verified that applying persistent homology to the superlevel sets of $f$ gives us the barcode $-S(\PH(G,-f))$. Applying persistent homology to the superlevel sets composed with the map $-\infty \mapsto \min(f)$ which changes the infinite endpoint to be closed is then easily seen to be $-S(\BPH(G,-f))$.

\subsection{Metrics on the space of persistence diagrams}

We first recall the definition of a persistence diagram.

\begin{definition}
A \textbf{persistence diagram} is a multiset of points in $\bar{\mathbb{R}}^2$ that contains every diagonal point $(x,x)$ for $x\in \mathbb{R}$ with infinite multiplicity. Here, $\bar{\mathbb{R}}$ denotes the extended real line $\mathbb{R} \cup \{-\infty\} \cup \{\infty\}$.
\end{definition}

Every barcode induces a persistence diagram by the map which sends each interval with left endpoint $x$ and right endpoint $y$ to the point $(x,y)$ and then includes every diagonal point $(x,x)$ with infinite multiplicity. This map is well defined even if $y < x$ for some subcollection of intervals.  We restrict our attention to persistence diagrams with finitely many off-diagonal points, each with finite multiplicity. Indeed, all finite trees $T$ must have persistence diagrams associated to $\TMD(T,f)$ and $\BPH(T,f)$ of this type since $T$ has finitely many nodes.

Given pairs of persistence diagrams $D_1$ and $D_2$, it is often useful to have a numerical measure of how different they are. Two popular such measurements of the difference between pairs of persistence diagrams are the $q$-Wasserstein and bottleneck distances.

\begin{definition}
For $q\geq 1$ the \textbf{$q$-Wasserstein distance} between persistence diagrams $D_1$ and $D_2$ is given by
\begin{equation*}
W_q(D_1,D_2) := \inf_{\phi} \Big{(} \sum_{x \in D_1} \lVert x - \phi(x)\rVert_\infty^q \Big{)}^{1/q}
\end{equation*}
where the infimum is taken over all bijective maps $\phi:D_1 \to D_2$. If the term inside the infimum is never finite then we let $W_q(D_1,D_2) = \infty$. We define the \textbf{bottleneck distance} similarly by
\begin{equation*}
d_B(D_1,D_2) = W_\infty(D_1,D_2) := \inf_{\phi} \sup_{x \in D_1} \lVert x - \phi(x)\rVert_\infty.
\end{equation*}
We say a bijective map $\phi:D_1\to D_2$ satisfying
\begin{equation*}
    \sup_{x \in D_1} \lVert x - \phi(x)\rVert_\infty \leq \delta
\end{equation*}
is a \textbf{$\delta$-matching} of $D_1$ and $D_2$.
\end{definition}
For this definition, we take $|\infty - \infty| = 0$ and $|\infty - t| = \infty$ for any $t\in\bar{\mathbb{R}}-\{\infty\}$, and similarly define the absolute value for $-\infty$.

The bottleneck distance is shown to be an extended metric on the space of persistence diagrams with finitely many off-diagonal points in \cite[p. 219]{edelsbrunner2010computational}. The Wasserstein distances are shown to be extended metrics similarly.

When comparing two barcodes $B_1$ and $B_2$ we abuse notation and let $W_q(B_1,B_2)$ and $d_B(B_1,B_2)$ denote the Wasserstein and bottleneck distances of their associated persistence diagrams. It should be noted that distance these functions are not metrics on the space of barcodes, for two barcodes can have the same persistence diagram and yet differ on the openness or closedness of their endpoints. As this information is all that is lost in the mapping from barcodes to persistence diagrams, this is the only way two barcodes of bottleneck or Wasserstein distance zero can differ.

\subsection{The TMD vs persistent homology}

The TMD algorithm also produces a barcode $\TMD(T,f)$ from a tree $T$ equipped with a function $f$ \cite{kanari2018topological}. We recall that the algorithm can be paraphrased as follows.
\begin{enumerate}
\item Choose a branch point, and identify a child branch of that branch point that attains the greatest value of $f$ on its leaves.
\item Detach every child branch from this branch point except for the identified branch.
\item If there are any branch points in the resulting collection of trees, return to step~1.
\item Label the endpoints of the resulting collection intervals with the $f$ values associated to their endpoints.
\end{enumerate}
Our goal is to connect the TMD operation to the $\BPH$ operation for a certain class of functions on $T$ including $d$, the intrinsic distance from the root. Prior related work, for example \cite{kanari2020trees} and \cite{li2017metrics}, has remarked and alluded that for functions $f$ which increase along paths moving away from the root, the TMD and $\BPH$ coincide. We provide a proof of this fact.

\begin{restatable}{theorem}{coresp}
Let $T$ be a finite rooted tree with root $r$, and $f$ be a function which is strictly increasing along the directed edges induced by $r$. We have that
\begin{equation*}
    \TMD(T,f) = -\BPH(T,-f).
\end{equation*}
\label{thm:coresp}
\end{restatable}
\begin{proof}

Let $l_0$ be one of the leaves with the largest $f$ value, and let $f(l_0) = L$. From the fact that $f$ is increasing with respect to the directed edges of $T$, it follows that the leaves of $T$ are the only local minima of $-f$, and so the intervals of $\PH(T,-f)$ bijectively correspond to leaves of $T$. Similarly, the TMD algorithm also associates an interval to each leaf of $T$. We show the theorem holds by computing $\TMD(T,f)$ and $\PH(T,-f)$ concurrently. For each iteration of the TMD algorithm, choose $b$, one of the branch points with the greatest value of $f$, in order to ensure that every detached child branch is actually an interval, and thus contains a single leaf. Then, choose one of the child branches which attains the largest value of $f$ on its leaf $l$. The Elder Rule \cite[Theorem 4.4]{cai2021elder} then dictates that each child branch containing a leaf $l'\neq l$ may be associated to the interval $[-f(l),-f(b))$ in $\PH(T,-f)$. Meanwhile the TMD algorithm associates to the child branch containing $l'\neq l$ the interval $(f(b),f(l')]$. The TMD algorithm then dictates that we detach the child branch containing each $l'\neq l$. This does not change the left endpoint given by the Elder Rule for an interval corresponding to any other leaf since removing the branch from $b$ to $l'$ does not change the minimal value of $f$ at any of the remaining child branches of branch points of $T$, nor the value of $f$ the remaining branch points of $T$ themselves. After completing this procedure for every branch point of $T$, what remains to be computed is the infinite interval $[-L,\infty)$ in $\PH(T,-f)$ and the finite interval $[f(r),L]$ in $\TMD(T,f)$. From the monotonicity assumption, $f(r) = \max(-f)$, and so it is immediate that $\TMD(T,f) = -\BPH(T,-f)$.
\end{proof}

An immediate consequence of this result is that the methods of \cite{kanari2018topological} and \cite{li2017metrics} can record identical information for functions $f$ satisfying the requirements of the theorem. In particular, the function $d$, the path distance to the root, satisfies the requirements of this theorem, and so $\TMD(T,d) = -\BPH(T,-d)$.

\subsection{Topological morphology functions from persistence surfaces}
For a neuron represented by a tree $T$, recall that the associated topological morphology function $p(t)$ returns the number of points on the neuron with an intrinsic distance of $t$ from the root.

We can transform the persistence diagram of $\TMD(T,f)$ by the map $(x,y)\mapsto (x,y-x)$. If $d$ is the intrinsic distance from the root function, we have shown in the main text that for $t$ not an endpoint of an interval of $\TMD(T,d)$, $p(t)$ is the number of points in transformed persistence diagram of $\TMD(T,d)$ in the region given by $x <t$ and $y>t-x$, which we call $R_t$. Adding together two dimensional Gaussian functions with standard deviation $\sigma$ centered at each point on the transformed persistence diagram of $\TMD(T,d)$, each weighted by the $y$ value of their center, we produce the persistence surface $F_\sigma(x,y)$. Theorem \ref{thm:limfunc}, which remains to be proven, shows that an approximation of $p$ can be constructed from $F_\sigma$.

We will formally restate this theorem shortly, but first we prove following lemma, which will be of use to us.

\begin{lemma}
\label{lem:distdelt}
Consider the family of boxes~$B_\delta(\mu)$ of width and height~$\delta$ centered around~$\mu = (\mu_x,\mu_y)$. Let~$g_\sigma(x,y;\mu)$ be a family of functions for positive $\sigma$ satisfying the following properties:
\begin{enumerate}
\item The function $g_\sigma(x,y;\mu)$ is positive, and is bounded for fixed $\sigma$.
\item $\int_{\mathbb{R}^2} g_\sigma(x,y;\mu)\,dx\,dy = 1$.
\item $g_\sigma(x,y;\mu) \to 0$ uniformly as $\sigma \to 0$ on $B_\delta^C(\mu)$, the complement of $B_\delta(\mu)$.
\item $\int_{B_\delta(\mu)} g_\sigma(x,y;\mu)\,dx\,dy \to 1$ as $\sigma \to 0$.
\item Every partial derivative of $g_\sigma(x,y;\mu)$ with respect to $x$ is continuous, bounded for fixed~$\sigma$, and the integrals of their absolute values over linear domains in $x$ and $y$ is bounded.
\end{enumerate}
Then
\begin{equation}
\lim_{\sigma \to 0} \int_{R_t}\frac{g_\sigma(x,y;\mu)}{y}\,dx\,dy
\label{eqn:glimfunc}
\end{equation}
is equal to zero if $\mu$ is in the exterior of $R_t$ and is equal to $\mu_y^{-1}$ if $\mu$ is in the interior of $R_t$. Further, for each positive $\sigma$,
\begin{equation*}
\int_{R_t}\frac{g_\sigma(x,y;\mu)}{y}\,dx\,dy
\end{equation*}
is infinitely differentiable as a function of $t$.
\end{lemma}

\begin{proof}
We only use property 5 to show the differentiability condition holds at the end. From properties 2 and 4 we immediately have that $\int_{B_\delta^C(\mu)} g_\sigma(x,y;\mu)\,dx\,dy \to 0$ as $\sigma \to 0$.

We first show that the integral in the theorem converges regardless of $t$. For positive~$h$ let $\Delta_h$ denote the triangular subset of $R_t$ of points $(x,y)$ additionally satisfying that $y<h$. Let $M$ be a bound for $g_\sigma(x,y;\mu)$ for a given $\sigma$. We have
\begin{align*}
&\int_{\Delta_h} \frac{g_\sigma(x,y;\mu)}{y}\,dx\,dy= \int_0^h \int_{t-y}^t \frac{g_\sigma(x,y;\mu)}{y} \,dx\,dy\leq \int_0^h \int_{t-y}^t \frac{M}{y}\,dx\,dy = Mh,\\
&\int_{R_t-\Delta_h} \frac{g_\sigma(x,y;\mu)}{y}\,dx\,dy \leq \int_{R_t-\Delta_h} \frac{g_\sigma(x,y;\mu)}{h}\,dx\,dy \leq \frac{1}{h}.
\end{align*}

To establish the limiting value in the theorem, first consider $g_\sigma(x,y;\mu)$ fixing $\mu$ in the exterior of $R_t$. For such $\mu$ there exists a~$\delta$ sufficiently small that $B_\delta(\mu)$ is disjoint from $R_t$. Let $S(\sigma)$ be the supremum of $g_\sigma(x,y;\mu)$ on $B_\delta^C(\mu)$ as a function of $\sigma$. Property 3 implies that this approaches 0 as $\sigma$ does. We have

\begin{align*}
&\int_{\Delta_h} \frac{g_\sigma(x,y;\mu)}{y}\,dx\,dy= \int_0^h \int_{t-y}^t \frac{g_\sigma(x,y;\mu)}{y} \,dx\,dy\leq \int_0^h \int_{t-y}^t \frac{S(\sigma)}{y}\,dx\,dy = hS(\sigma),\\
&\int_{R_t-\Delta_h} \frac{g_\sigma(x,y;\mu)}{y}\,dx\,dy \leq \int_{R_t-\Delta_h} \frac{g_\sigma(x,y;\mu)}{h}\,dx\,dy \leq \frac{1}{h}.
\end{align*}

Letting $h = S(\sigma)^{-1/2}$ we observe
\begin{equation*}
0 \leq \int_{R_t} \frac{g_\sigma(x,y;\mu)}{y}\,dx\,dy \leq 2S(\sigma)^{1/2},
\end{equation*}
which tends to zero as $\sigma$ does.

If on the other hand $\mu$ lies in $R_t$, again we can choose $\delta$ small enough that $B_\delta(\mu)$ is a subset of $R_t$. Define $S(\sigma)$ as before. Once again we obtain bounds 
\begin{align*}
&\int_{\Delta_h - B_\delta(\mu)} \frac{g_\sigma(x,y;\mu)}{y}\,dx\,dy \leq \int_{\Delta_h} \frac{S(\sigma)}{y}\,dx\,dy = \int_0^h \int_{t-y}^t \frac{S(\sigma)}{y}\,dx\,dy = hS(\sigma),\\
&\int_{R_t-\Delta_h\cup B_\delta(\mu)} \frac{g_\sigma(x,y;\mu)}{y}\,dx\,dy \leq \int_{R_t-\Delta_h} \frac{g_\sigma(x,y;\mu)}{h}\,dx\,dy \leq \frac{1}{h}.
\end{align*}
Combining these two results, and letting $h = S(\sigma)^{-1/2}$ we observe
\begin{equation*}
0 \leq \int_{R_t - B_\delta(\mu)} \frac{g_\sigma(x,y;\mu)}{y}\,dx\,dy \leq 2S(\sigma)^{1/2},
\end{equation*}
which approaches zero as $\sigma$ approaches zero. Meanwhile, we also have
\begin{align*}
&\int_{B_\delta(\mu)} \frac{g_\sigma(x,y;\mu)}{y}\,dx\,dy \leq \int_{B_\delta(\mu)} \frac{g_\sigma(x,y;\mu)}{\mu_y - \delta}\,dx\,dy \leq \frac{1}{\mu_y - \delta}\\
&\int_{B_\delta(\mu)} \frac{g_\sigma(x,y;\mu)}{y}\,dx\,dy \geq \int_{B_\delta(\mu)} \frac{g_\sigma(x,y;\mu)}{\mu_y + \delta}\,dx\,dy \to \frac{1}{\mu_y + \delta},
\end{align*}
with the limit in the last line taken as $\sigma \to 0$. Since $\delta$ can be taken arbitrarily small, this implies
\begin{equation*}
\int_{B_\delta(\mu)}  \frac{g_\sigma(x,y;\mu)}{y}\,dx\,dy \to \frac{1}{\mu_y}.
\end{equation*}
Hence,
\begin{equation*}
\int_{R_t}  \frac{g_\sigma(x,y;\mu)}{y}\,dx\,dy \to \frac{1}{\mu_y}.
\end{equation*}
In summary, we have shown that the above integral approaches $\mu_y^{-1}$ when $\mu$ is interior to $R_t$ and approaches zero when $\mu$ is exterior to $R_t$.

All that remains to be shown is the differentiability statement of the lemma, i.e. we must show that the integral
\begin{equation*}
I(t;\mu) := \int_{R_t}\frac{g_\sigma(x,y;\mu)}{y}\,dx\,dy = \int_0^\infty \int_{t-y}^t \frac{g_\sigma(x,y;\mu)}{y}\,dx\,dy 
\end{equation*}
is an infinitely differentiable function of t. 

To begin, notice that by continuity of $g$ we can extend the integrand to $y=0$ via
\begin{equation*}
\lim_{y \to 0}  \int_{t-y}^t \frac{g_\sigma(x,y;\mu)}{y}\,dx = g_\sigma (t,y;\mu).
\end{equation*}
Further, we have the estimate for positive $C$:
\begin{equation*}
\int_C^\infty  \int_{t-y}^t \frac{g_\sigma(x,y;\mu)}{y}\,dx \, dy \leq \int_C^\infty  \int_{t-y}^t \frac{g_\sigma(x,y;\mu)}{C}\,dx \, dy \leq \int_{\mathbb{R}^2}\frac{g_\sigma(x,y;\mu)}{C}\,dx \, dy \leq \frac{1}{C},
\end{equation*}
showing uniform convergence of the same integral with $C$ taken to be 0, since $g$ is positive valued.
We also have that
\begin{align*}
\frac{\partial^n}{\partial t^n}   \int_{t-y}^t \frac{g_\sigma(x,y;\mu)}{y}\,dx \, dy &= \frac{\partial^{n-1}}{\partial t^{n-1}}  \frac{g_\sigma(t,y;\mu)- g_\sigma(t-y,y;\mu)}{y} \\
&= \frac{1}{y}\Big{[}\frac{\partial^{n-1}}{\partial t^{n-1}}g_\sigma(t,y;\mu)- \frac{\partial^{n-1}}{\partial t^{n-1}}g_\sigma(t-y,y;\mu)\Big{]}.
\end{align*}
This expression is bounded by any constant bounding of $|\frac{\partial^n}{\partial t^n}g_\sigma(t,y;\mu)|$. Hence, applying $\int_0^\infty dy$ to this expression, the lower limit converges uniformly. Meanwhile, letting $M$ be a bound for the integral of $|\frac{\partial^{n-1}}{\partial t^{n-1}}g_\sigma(t,y;\mu)|$ over linear domains in $t$ and $y$, we observe that the upper limit converges uniformly as well since
\begin{align*}
&\bigg{|}\int_C^\infty \frac{1}{y}\Big{[}\frac{\partial^{n-1}}{\partial t^{n-1}}g_\sigma(t,y;\mu)-  \frac{\partial^{n-1}}{\partial t^{n-1}}g_\sigma(t-y,y;\mu)\Big{]}\,dy \bigg{|}\\
&\leq \int_C^\infty \frac{1}{C} \bigg{[}\Big{|}\frac{\partial^{n-1}}{\partial t^{n-1}}g_\sigma(t,y;\mu)\Big{|} + \Big{|} \frac{\partial^{n-1}}{\partial t^{n-1}}g_\sigma(t-y,y;\mu)\Big{|}\bigg{]}\,dy \\
&\leq\frac{2M}{C}.
\end{align*}
From this we know we can interchange the partial derivative and integral in the expression for $I(t;\mu)$, giving us the formula
\begin{equation*}
\frac{\partial^n}{\partial t^n} I(t;\mu) = \int_0^\infty \frac{1}{y}\Big{[}\frac{\partial^{n-1}}{\partial t^{n-1}}g_\sigma(t,y;\mu)-  \frac{\partial^{n-1}}{\partial t^{n-1}}g_\sigma(t-y,y;\mu)\Big{]}\,dy,
\end{equation*}
completing the proof.
\end{proof}

With this Lemma we can easily show Theorem \ref{thm:limfunc}.

\limfunc*

\begin{proof}[Proof of Theorem \ref{thm:limfunc}]
Let $\mu_1,\ldots,\mu_N$ be the coordinates of the off-diagonal points in the transformed persistence diagram of $\TMD(T,f)$. The persistence surface is then given by the formula
\begin{equation}
F_\sigma(x,y) = \sum_{i=1}^N (\mu_i)_y \, g_\sigma(x,y;\mu_i).
\label{eqn:fdef}
\end{equation}
Each function $g_\sigma(x,y;\mu_i)$ is a two dimensional Gaussian function, for which the requirements of Lemma~\ref{lem:distdelt} are elementary properties. Each statement of Theorem~\ref{thm:limfunc} is immediate from the fact that $F_\sigma$ is a finite linear combination of functions satisfying the requirements of Lemma~\ref{lem:distdelt}. Indeed, Equation~\ref{eqn:fdef} is clearly infinitely differentiable from the previous lemma. When $t$ is generic with respect to $B$, it is easily seen that each $\mu_i$ is not on the boundary of $R_t$. Applying Lemma \ref{lem:distdelt} and Equation \ref{eqn:fdef}, we see that the limit
\begin{equation*}
    \lim_{\sigma \to 0} \int_{R_t}\frac{F_\sigma(x,y)}{y}\,dx\,dy
\end{equation*}
evaluates to the number of $\mu_i$ in $R_t$, which is equal to $p(t)$.
\end{proof}

Let $T$ be a finite tree, with an embedding $P$ to $\mathbb{R}^n$, and $f$ be the Euclidean distance of each point represented by $v$ to the root $r$.  Recall that we define the Sholl function $s(t)$ to be the number times the $n$-sphere of radius $t$ centered about $P(r)$ intersects with $P(T)$. Let $D$ be persistence diagram associated to the barcode
\begin{equation}
    \BPH(T,f) \sqcup -S(\BPH(T,-f)),
\label{eqn:Ddef}
\end{equation}
the disjoint union of barcodes given by persistent homology of the superlevel and sublevel sets of $f$. We say that $t$ is generic if it is not an endpoint of an interval in $D$. If instead of assuming the embedding of edges of $T$ in $\mathbb{R}^n$ is linear, we assume that $P$ is such that $f$ is linear on edges, Li et al. show in \cite[section 2.4]{li2017metrics} that the value $s(t)$ can be recovered for any generic $t$. For the remainder of the appendix, we only consider $T$ with such an embedding. For generic $t$ between 0 and $\max(f)$, the Sholl function $s(t)$ is the number of points above and to the left or below and to the right of $(t,t)$ minus one. More formally written, if $Q_t$ is the set of $(x,y)$ such that ($x < t$ and $y > t$) or ($x > t$ and $y < t$), then
\begin{equation}
s(t) = \big|D\cap Q_t \big| -1.
\end{equation}
If $t$ is generic, but $t$ is not between 0 and $max(f)$ it is easily seen that the left side of this equation is 0 while the right side of this equation is -1. We show how $s(t)$ may be calculated from $D$ visually in Figure \ref{fig:lisholl}.

\begin{figure}[htbp]
\centering
\resizebox{0.5\textwidth}{!}{
\includegraphics{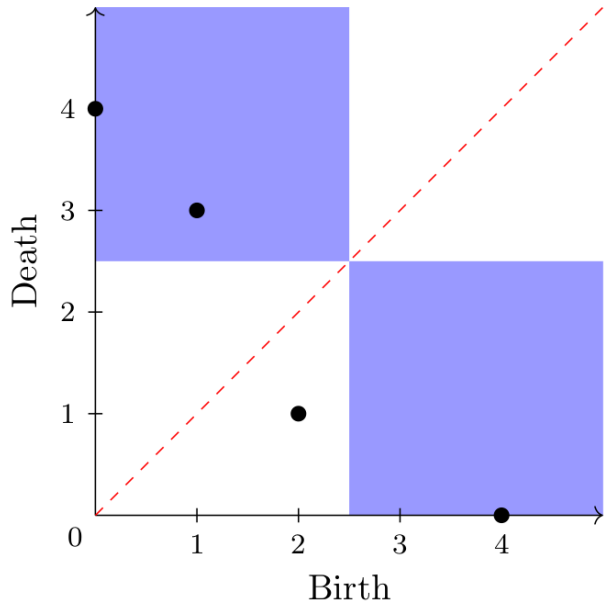}
}

\caption{Calculating the $s(2.5)$ from the persistence diagram $D$. Shown in blue is the region $Q_{2.5}$, within which there are three points. Hence, $s(2.5) = 2$.}
\label{fig:lisholl}
\end{figure}

One may similarly define a persistence surface $F_\sigma$ for a barcode $D$ by adding a two dimensional Gaussian function for each point in $D$, with multiplicity, weighted by the vertical distance of each point from the diagonal\footnote{It is common to skip the transformation step when there are points below the diagonal in the initial barcode.}. We can use our Lemma \ref{lem:distdelt} to show that an approximate Sholl function can be recovered from this persistence surface.

\begin{theorem}
Let $T$ be a rooted tree with Sholl function $s$ and $f$ the associated Euclidean distance from the root function. Assume that between adjacent vertices, the embedding of $T$ in $\mathbb{R}^2$ or $\mathbb{R}^3$ is such that $f$ is linear. Let $D$ be the persistence diagram of the barcode given by Equation \ref{eqn:Ddef}. Let $F_\sigma(x,y)$ be the persistence surface corresponding to $D$ constructed with Gaussian functions of standard deviation $\sigma$. For any generic positive number $t$ between 0 and $\max(f)$,

\begin{equation*}
s(t) = \lim_{\sigma \to 0} \int_{Q_t}\frac{F_\sigma(x,y)}{|x-y|}\,dx\,dy-1.
\end{equation*}
Otherwise, if $t$ is generic, the left side of this equation is 0 and the right side of this expression is -1. Further, the integral in the above expression converges and is an infinitely differentiable function of $t$ for all $\sigma > 0$.
\label{thm:limfuncsholl}
\end{theorem}

\begin{proof}
Let $\mu_1,\ldots,\mu_N$ be the coordinates of each point in $D$, with $x$ and $y$ values $(\mu_i)_x$ and~$(\mu_i)_y$ respectively. By definition the persistence surface $F_\sigma$ is given by a weighted sum of Gaussians:
\begin{equation*}
F_\sigma(x,y) = \sum_{i=1}^N \big{|}(\mu_i)_y -(\mu_i)_x\big{|}\, g_\sigma(x,y;\mu_i).
\label{eqn:fdefli}
\end{equation*}
We denote by $Q_t^+$ the portion of $Q_t$ satisfying $y > x$ and similarly denote by $Q_t^-$ the portion of $Q_t$ satisfying $x > y$. We have the equation
\begin{equation*}
 \int_{Q_t}\frac{F_\sigma(x,y)}{|x-y|}\,dx\,dy =  \int_{Q_t^+}\frac{F_\sigma(x,y)}{y-x}\,dx\,dy +  \int_{Q_t^-}\frac{F_\sigma(x,y)}{x-y}\,dx\,dy.
\end{equation*}
Examining the first term, we have
\begin{equation*}
\int_{Q_t^+}\frac{F_\sigma(x,y)}{y-x}\,dx\,dy =\int_0^\infty  \int_{-\infty}^t\frac{F_\sigma(x,y)}{y-x}\,dx\,dy = \int_0^\infty \int_{t-v}^{t} \frac{F_\sigma(u,u+v)}{v}\,du\,dv
\end{equation*}
 via the transform $(x,y)\mapsto(x,y-x):= (u,v)$. Noting that linear transformations of Gaussian density functions are still Gaussian density functions and that this last integral is over the region $R_t$ in the $uv$ plane, we have:
\begin{equation*}
\int_0^\infty \int_{t-v}^{t} \frac{F_\sigma(u,u+v)}{v}\,du\,dv = \sum_{i=1}^N \big{|}(\mu_i)_y -(\mu_i)_x\big{|}\int_{R_t}\frac{g_\sigma(u,u+v;\mu_i)}{v}\,du\,dv.
\end{equation*}
We may now apply Lemma \ref{lem:distdelt} as each $g_\sigma$ is Gaussian with mean $((\mu_i)_x,(\mu_i)_y - (\mu_i)_x)$ in the $uv$ plane and therefore upholds the requirements of the lemma. Thus the integrals on the right are well defined and infinitely differentiable. Consider the $i^\mathrm{th}$ of these integrals on the right. Taking the limit $\sigma \to 0$ this integral evaluates to $((\mu_i)_y - (\mu_i)_x)$ if $((\mu_i)_x,(\mu_i)_y - (\mu_i)_x)$ is interior to $R_t$ and 0 it is exterior to $R_t$. Equivalently, the integral evaluates to $((\mu_i)_y - (\mu_i)_x)$ if $\mu_i$ is interior to $Q_t^+$ and 0 if it is exterior to $Q_t^+$.

No $\mu_i$ lies on the boundary of $Q_t^+$ when $t$ is generic. Hence, the limit
\begin{equation*}
\lim_{\sigma \to 0} \int_{Q_t^+}\frac{F_\sigma(x,y)}{|x-y|}\,dx\,dy
\end{equation*}
evaluates to the number of $\mu_i$ in $Q_t^+$ for generic $t$.

Similarly, we may show via the transformation $(x,y) \mapsto (y, x-y)$ that the limit
\begin{equation*}
\lim_{\sigma \to 0} \int_{Q_t^-}\frac{F_\sigma(x,y)}{|x-y|}\,dx\,dy
\end{equation*}
is the number of $\mu_i$ in $Q_t^-$ for generic $t$, and that the integral here is always well defined and infinitely differentiable regardless of $t$. 

Hence the integral
\begin{equation*}
\int_{Q_t}\frac{F_\sigma(x,y)}{|x-y|}\,dx\,dy =\int_{Q_t^+}\frac{F_\sigma(x,y)}{|x-y|}\,dx\,dy + \int_{Q_t^-}\frac{F_\sigma(x,y)}{|x-y|}\,dx\,dy
\end{equation*}
is well defined and infinitely differentiable for any $\sigma$.

Now consider generic $t$. We have 
\begin{equation*}
\lim_{\sigma \to 0} \int_{Q_t}\frac{F_\sigma(x,y)}{|x-y|}\,dx\,dy-1 =\lim_{\sigma \to 0} \int_{Q_t^+}\frac{F_\sigma(x,y)}{|x-y|}\,dx\,dy + \lim_{\sigma \to 0} \int_{Q_t^-}\frac{F_\sigma(x,y)}{|x-y|}\,dx\,dy - 1,
\end{equation*}
which is the number of points in the barcode that also lie in $Q_t$ minus one.
This value has been shown to be $s(t)$ in \cite{li2017metrics} when $t$ is between 0 and $\max(f)$. When $t$ is not between 0 and $\max(f)$, there can be no points in $Q_t^+$ or $Q_t^-$, and so this expression evaluates to -1. Clearly for such $t$, $s(t) = 0$.

\end{proof}

\end{document}